\documentclass[11pt,a4paper]{article}
\usepackage{indentfirst,mathrsfs}
\usepackage{amsfonts,amsmath,amssymb,amsthm}
\usepackage{latexsym,amscd}
\usepackage{amsbsy}
\usepackage[all]{xy}
\usepackage{graphicx}
\usepackage{threeparttable}
\usepackage[usenames,dvipsnames,svgnames,table]{xcolor}
\usepackage{cite}

\newtheorem{theorem}{Theorem}            

\newtheorem{thm}{Theorem}
\newtheorem{lemma}{Lemma}

\newtheorem{Example}{Example}

\begin{document}
\date{}
\title{Constructions of MDS symbol-pair codes with minimum distance seven or eight}
\maketitle
\begin{center}
\author{\large Junru Ma \quad\quad Jinquan Luo
\footnote{Corresponding author
\par
Junru Ma is with Faculty of Mathematics and Statistics, Hubei Key Laboratory of Applied Mathematics, Hubei University, Wuhan 430062, Hubei, China.
\par
Jinquan Luo is with School of Mathematics and Statistics, Hubei Key Laboratory of Mathematical Sciences, Central China Normal University, Wuhan 430079, Hubei, China.\\
\par
E-mails: majunru@hubu.edu.cn(J.Ma), luojinquan@mail.ccnu.edu.cn(J.Luo).
}}
\end{center}

\begin{quote}
  {\small {\bf Abstract:} \ \
  Symbol-pair codes are proposed to guard against pair-errors in symbol-pair read channels.
  The minimum symbol-pair distance plays a vital role in determining the error-correcting capability and the constructions of symbol-pair codes with largest possible minimum symbol-pair distance is of great importance.
  Maximum distance separable (\,MDS\,) symbol-pair codes are optimal in the sense that such codes can acheive the Singleton bound.
  In this paper, for length $5p$, two new classes of MDS symbol-pair codes with minimum symbol-pair distance seven or eight are constructed by utilizing repeated-root cyclic codes over $\mathbb{F}_{p}$, where $p$ is a prime.
  In addition, we derive a class of MDS symbol-pair codes with minimum symbol-pair distance seven and length $4p$.
  }

  {\small {\bf Keywords:} \ \ MDS symbol-pair codes, \ minimum symbol-pair distance, \ constacyclic codes, \ repeated-root cyclic codes}
\end{quote}

\section{Introduction}

With the development of modern high density data storage systems, symbol-pair code was proposed by Cassuto and Blaum to combat against pair-errors over symbol-pair read channels in \cite{CB1,CB2}.
They also showed that a code $\mathcal{C}$ with minimum symbol-pair distance $d_{p}$ can correct up to $\left\lfloor(d_{p}-1)/2\right\rfloor$ symbol-pair errors \cite{CB1,CB2}.
Later, Cassuto and Litsyn \cite{CL} showed that codes for correcting pair-errors exist with strictly higher rates compared to codes for the Hamming metric with the same relative distance.
In \cite{CKW}, Chee, Kiah and Wang established a Singleton-type bound on symbol-pair codes.
Similar to classical codes, symbol-pair codes meeting the Singleton-type bound are called MDS symbol-pair codes
and the error-correcting capability of MDS symbol-pair codes is optimal.
Later, Ding, Zhang and Ge extended the Singleton-type bound to the $b$-symbol case in \cite{DZGb}.

Many attempts have been made in the constructions of MDS symbol-pair codes.
In \cite{KZL}, Kai, Zhu and Li provided MDS symbol-pair codes with length $q^2+q+1$ through constacyclic codes over $\mathbb{F}_{q}$.
Later, Li and Ge \cite{LG} generalized the results in \cite{KZL} and they also constructed a number of MDS symbol-pair codes with minimum symbol-pair distance seven by analyzing certain linear fractional transformations.
Shortly afterwards, Chen, Lin and Liu \cite{CLL} constructed several MDS symbol-pair codes with length $3p$ from repeated-root cyclic codes over $\mathbb{F}_{p}$.
In $2018$, Ding et al. \cite{DGZZZ} obtained some MDS symbol-pair codes over $\mathbb{F}_{q}$ with larger minimum symbol-pair distance based on elliptic curves and the lengths of these codes are bounded by $q+2\sqrt{q}$.
In the same year, Kai et al. \cite{KZZLC} constructed three classes of MDS symbol-pair codes using repeated-root constacyclic codes over $\mathbb{F}_{p}$, see Table \ref{Tab-constacyclic-MDS}.
Recently, some new results on constructing symbol-pair codes were presented in \cite{DNS,EGY,ML1}.
Moreover, some decoding algorithms of symbol-pair codes were proposed by various researchers in \cite{HMH,LXY,MHT,THM,YBS} and the symbol-pair weight distributions of some linear codes over finite fields were studied in \cite{DNSS1,DNSS2,DWLS,ML2,SZW} and the references therein.

In Table \ref{Tab-constacyclic-MDS}, we summarize some known MDS symbol-pair codes from constacyclic codes.
\begin{table}[!htb]
  \caption{Some known MDS symbol-pair codes from constacyclic codes}\label{tab:tablenotes}
  \centering
  \begin{threeparttable}\label{Tab-constacyclic-MDS}
  \begin{tabular}{|c|c|c|}
  \hline Values of $(n,\,d_{p})_{q}$      &Conditions                &References\\
  \hline $\left(n,\,5\right)_{q}$       &$n\,|\left(q^{2}+q+1\right)$
    &\cite{KZL},\cite{LG} \\
  \hline $\left(n,\,6\right)_{q}$       &$n\,|\left(q^{2}+1\right)$
    &\cite{KZL},\cite{LG} \\
  \hline $\left(n,\,6\right)_{q}$
    &$n\,|\left(q^{2}-1\right)$,\,$n$\,odd or $n$ even and $v_{2}\left(n\right)<v_{2}\left(q^{2}-1\right)$     &\cite{LG} \\
  \hline $\left(n,\,6\right)_{q}$
    &$q\geq3,\,n\geq q+4,\,n\,|\left(q^{2}-1\right)$    &\cite{CLL} \\
  \hline $\left(lp,\,5\right)_{p}$
    &$p\geq5,\,l>2,\,{\rm gcd}\left(l,\,p\right)=1,\,l\,|\left(p-1\right)$
    &\cite{CLL} \\
  \hline $\left(p^{2}+p,\,6\right)_{p}$    &$p\geq3$          &\cite{KZZLC} \\
  \hline $\left(2p^{2}-2p,\,6\right)_{p}$  &$p\geq3$          &\cite{KZZLC} \\
  \hline $\left(3p,\,6\right)_{p}$         &$p\geq5$          &\cite{CLL}   \\
  \hline $\left(3p,\,7\right)_{p}$         &$p\geq5$          &\cite{CLL}   \\
  \hline $\left(3p,\,8\right)_{p}$         &$3\,|\left(p-1\right)$ &\cite{CLL} \\
  \hline $\left(3p,\,10\right)_{p}$        &$3\,|\left(p-1\right)$     &\cite{ML1}   \\
  \hline $\left(3p,\,12\right)_{p}$        &$3\,|\left(p-1\right)$     &\cite{ML1} \\
  \hline $\left(4p,\,7\right)_{p}$         &$p\equiv 3\left({\rm mod}\,4\right)$
    &\cite{KZZLC} \\
  \hline $\left(4p,\,7\right)_{p}$         &$p\equiv 1\left({\rm mod}\,4\right)$
    &Theorem \ref{thmMDS4p7p} \\
  \hline $\left(5p,\,7\right)_{p}$         &$5\,|\left(p-1\right)$, $p\ne 41$
    &Theorem \ref{thmMDS5p7p} \\
  \hline $\left(5p,\,8\right)_{p}$         &$5\,|\left(p-1\right)$
    &Theorem \ref{thmMDS5p8p} \\
  \hline
  \end{tabular}
  \begin{tablenotes}
        \footnotesize
          \item  where $q$ is a power of a prime $p$.
  \end{tablenotes}
  \end{threeparttable}
\end{table}
Observe that there exists only one class of codes with length $5p$ and minimum symbol-pair distance five in Table \ref{Tab-constacyclic-MDS}.
The constructions of symbol-pair codes with comparatively large minimum symbol-pair distance is an interesting topic.
This paper focuses on the further constructions of MDS symbol-pair codes with length $5p$.
Precisely, two new classes of MDS symbol-pair codes with minimum symbol-pair distance seven or eight are constructed by utilizing repeated-root cyclic codes over $\mathbb{F}_{p}$.
In addition, for $n=4p$, we derive a class of MDS symbol-pair codes with $d_{p}=7$, which generalizes the result in \cite{KZZLC}.

The remainder of this paper is organized as follows.
In Section $2$, we introduce some basic notation and results on symbol-pair codes and constacyclic codes.
By exploiting repeated-root cyclic codes, for length $5p$, two new classes of MDS symbol-pair codes with minimum symbol-pair distance seven or eight are constructed in Section $3.1$ and a class of MDS symbol-pair codes with length $4p$ is presented in Section $3.2$.
Section $4$ concludes the paper.

\section{Preliminaries}

In this section, we introduce some notation and auxiliary tools on symbol-pair codes and constacyclic codes, which will be used to prove our main results in the sequel.

\subsection{Symbol-pair codes}

Let $\mathbb{F}_{q}$ be the finite field with $q$ elements, where $q$ is a prime power.
Let $n$ be a positive integer and $\mathbf{x}=\left(x_{0}, x_{1}, \cdots, x_{n-1}\right)$ be a vector in $\mathbb{F}_{q}^{n}$.
Then the {\it symbol-pair read vector} of $\mathbf{x}$ is
\begin{equation*}
  \pi(\mathbf{x})=\left(\left(x_{0},\,x_{1}\right),
  \,\left(x_{1},\,x_{2}\right),\cdots,\left(x_{n-2},\,x_{n-1}\right),
  \,\left(x_{n-1},\,x_{0}\right)\right).
\end{equation*}
Obviously, each vector $\mathbf{x}\in \mathbb{F}_{q}^{n}$ has a unique pair representation $\pi(\mathbf{x})$.
Recall that the {\it Hamming weight} of $\mathbf{x}$ is
\begin{equation*}
  w_{H}(\mathbf{x})=\left|\left\{i\in \mathbb{Z}_{n}\,\big|\,x_{i}\neq 0\right\}\right|
\end{equation*}
where $\mathbb{Z}_{n}$ denotes the residue class ring $\mathbb{Z}/n\mathbb{Z}$.
Correspondingly, the {\it symbol-pair weight} of $\mathbf{x}$ is
\begin{equation*}
  w_{p}(\mathbf{x})=\left|\left\{i\in \mathbb{Z}_{n}\,\big|\,\left(x_{i},\,x_{i+1}\right)\neq \left(0,\,0\right)\right\}\right|.
\end{equation*}
For any two vectors $\mathbf{x},\,\mathbf{y}\in \mathbb{F}_{q}^{n}$, the {\it symbol-pair distance} from $\mathbf{x}$ to $\mathbf{y}$ is defined as
\begin{equation*}
  d_{p}\left(\mathbf{x},\,\mathbf{y}\right)=\left|\left\{i\in \mathbb{Z}_{n}\,\big|\,\left(x_{i},\,x_{i+1}\right)\neq \left(y_{i},\,y_{i+1}\right)\right\}\right|.
\end{equation*}
A code $\mathcal{C}$ over $\mathbb{F}_{q}$ of length $n$ is a nonempty subsets of $\mathbb{F}_{q}^{n}$.
Elements of $\mathcal{C}$ are called {\it codewords} in $\mathcal{C}$.
The {\it minimum symbol-pair distance} of $\mathcal{C}$ is
\begin{equation*}
  d_{p}(\mathcal{C})={\rm min}\left\{d_{p}\left(\mathbf{x},\,\mathbf{y}\right)\,\big|\,\, \mathbf{x},\,\mathbf{y}\in \mathcal{C},\mathbf{x}\neq\mathbf{y}\right\}
\end{equation*}
and we refer such a code as an $\left(n,\,d_{p}(\mathcal{C})\right)_{q}$ {\it symbol-pair code}.
A well-known relationship between $d_{H}(\mathcal{C})$ and $d_{p}(\mathcal{C})$ in \cite{CB1,CB2} states that for any $0<d_{H}(\mathcal{C})<n$,
$$d_{H}(\mathcal{C})+1\leq d_{p}(\mathcal{C})\leq 2\cdot d_{H}(\mathcal{C}).$$

The following lemma reveals a connection between the symbol-pair distance and the Hamming distance of a code $\mathcal{C}$.

\begin{lemma}{\rm(\cite{CB1,CB2})}\label{lemdpdH}
For any $\mathbf{x},\mathbf{y}\in \mathcal{C}$ with $\mathbf{x}=(x_0,\cdots,x_{n-1})$ and $\mathbf{y}=(y_0,\cdots,y_{n-1})$.
Define $S=\{i\in \mathbb{Z}_{n}\,|\,x_i\ne y_i\}$.
Let $S=\bigcup _{i=1}^{L}S_i$ be a partition of $S$, which satisfies:

(1) the elements of each subset $S_i$ are consecutive in the sense of modulo $n$;

(2) for any different $i,\,j\in [1,L]$ and $a\in S_{i},\,b\in S_{j}$, $a$ and $b$ are not consecutive.
\\
Then
\begin{equation*}
  d_{p}\left(\mathbf{x},\,\mathbf{y}\right)=d_{H}\left(\mathbf{x},\,\mathbf{y}\right)+L.
\end{equation*}
\end{lemma}

In contrast to classical error-correcting codes, the size of symbol-pair codes satisfies the following Singleton bound.

\begin{lemma}{\rm(\cite{CJKWY})}\label{Singleton}
Let $q\geq 2$ and $2\leq d_{p}\leq n$.
If $\mathcal{C}$ is a symbol-pair code with length $n$ and minimum symbol-pair distance $d_{p}$, then $\left|\mathcal{C}\right|\leq q^{n-d_{p}+2}$.
\end{lemma}

The symbol-pair code achieving the Singleton bound is called a {\it maximum distance separable} (\,MDS\,) symbol-pair code.

\subsection{Constacyclic codes}

In this subsection, we introduce some notation of constacyclic codes.
For a fixed nonzero element $\eta$ of $\mathbb{F}_{q}$, the $\eta$-constacyclic shift $\tau_{\eta}$ on $\mathbb{F}_{q}^{n}$ is
\begin{equation*}
  \tau_{\eta}\left(x_{0},\,x_{1},\cdots,x_{n-1}\right)
  =\left(\eta\,x_{n-1},\,x_{0},\cdots,x_{n-2}\right).
\end{equation*}
A linear code $\mathcal{C}$ is called an {\it$\eta$-constacyclic code} if $\tau_{\eta}\left(\mathbf{c}\right)\in\mathcal{C}$ for any codeword $\mathbf{c}\in\mathcal{C}$.
An $\eta$-constacyclic code is a {\it cyclic code} if $\eta=1$ and a {\it negacyclic code} if $\eta=-1$.
It should be noted that each codeword $\mathbf{c}=\left(c_{0},\,c_{1},\cdots,c_{n-1}\right)\in\mathcal{C}$ is identical to its polynomial representation
\begin{equation*}
  c(x)=c_{0}+c_{1}\,x+\cdots+c_{n-1}\,x^{n-1}.
\end{equation*}
For convenience, we always regard the codeword $\mathbf{c}$ in $\mathcal{C}$ as the corresponding polynomial $c(x)$ in this paper.
Notice that a linear code $\mathcal{C}$ is an $\eta$-constacyclic code if and only if it is an ideal of the principle ideal ring $\mathbb{F}_{q}[x]/\langle x^{n}-\eta\rangle$.
As a consequence, there exists a unique monic divisor $g(x)\in\mathbb{F}_{q}[x]$ of $x^{n}-\eta$ such that
\begin{equation*}
  \mathcal{C}=\left\langle g(x)\right\rangle=\left\{f(x)\,g(x)\left({\rm mod}\left(x^n-\eta\right)\right)\big|\,f(x)\in\mathbb{F}_{q}\left[x\right]\right\}.
\end{equation*}
The polynomial $g(x)$ is called the {\it generator polynomial} of $\mathcal{C}$ and the dimension of $\mathcal{C}$ is $n-k$, where $k$ is the degree of $g(x)$.

Recall that a $q$-ary $\eta$-constacyclic code of length $n$ is a {\it simple-root} constacyclic code if ${\rm gcd}\left(n,\,q\right)=1$ and a {\it repeated-root} constacyclic code if $p\,|\,n$, where $p$ is the characteristic of $\mathbb{F}_{q}$.
Simple-root constacyclic codes can be characterized by their defining sets \cite{HP,MS}.
Compared to simple-root cyclic codes, repeated-root cyclic codes are no longer characterized by its set of zeros.
Let $\mathcal{C}=\left\langle g(x)\right\rangle$ be a repeated-root cyclic code of length $lp^{e}$ over $\mathbb{F}_{q}$, where $l$ and $e$ are positive integers with ${\rm gcd}\left(l,\,p\right)=1$.
It is shown in \cite{CMSS} that the minimum Hamming distance of $\mathcal{C}$ can be derived from $d_{H}(\overline{\mathcal{C}}_{t})$.
Here $\overline{\mathcal{C}}_{t}$ is a simple-root cyclic code fully determined by $\mathcal{C}$ as follows.

More precisely, assume that
\begin{equation*}
  g(x)=\prod_{i=1}^{r}m_{i}(x)^{e_{i}}
\end{equation*}
where each $m_{i}(x)$ is a monic irreducible polynomial over $\mathbb{F}_{q}$ and $e_{i}$ are positive integers.
For a fixed $t$ with $0\leq t\leq p^{e}-1$, $\overline{\mathcal{C}}_{t}$ is defined to be a simple-root cyclic code of length $l$ over $\mathbb{F}_{q}$ with the generator polynomial
\begin{equation*}
  \overline{g}_{t}(x)=\prod_{1\leq i\leq r,\,e_{i}>t}m_{i}(x).
\end{equation*}
If $\overline{g}_{t}(x)=x^{l}-1$, then $\overline{\mathcal{C}}_{t}$ contains only the all-zero codeword and we set $d_{H}(\overline{\mathcal{C}}_{t})=\infty$.
If each $e_{i}\leq t$, then  $\overline{g}_{t}(x)=1$ and $d_{H}(\overline{\mathcal{C}}_{t})=1$.

The following lemma reveals that the minimum Hamming distance of repeated-root cyclic codes can be determined by the polynomial algebra, which will be applied to derive the minimum Hamming distance of codes in this paper.

\begin{lemma}{\rm(\cite{CMSS})}\label{lemdistance}
Let $\mathcal{C}$ be a repeated-root cyclic code of length $lp^{e}$ over $\mathbb{F}_{q}$, where $l$ and $e$ are positive integers with ${\rm gcd}\left(l,\,p\right)=1$.
Then
\begin{equation}\label{eqdistance}
  d_{H}(\mathcal{C})={\rm min}\left\{P_{t}\cdot d_{H}\left(\overline{\mathcal{C}}_{t}\right)\,\big|\,0\leq t\leq p^{e}-1\right\}
\end{equation}
where
\begin{equation}\label{eqpt}
  P_{t}=w_{H}\left(\left(x-1\right)^{t}\right)=\prod_{i}\left(t_{i}+1\right)
\end{equation}
with $t_{i}$'s being the coefficients of the radix-$p$ expansion of $t$.
\end{lemma}

In this paper, we will employ repeated-root cyclic codes to construct new MDS symbol-pair codes.  The following lemmas are very useful.

\begin{lemma}{\rm(\cite{CLL})}\label{lemMDS}
Let $\mathcal{C}$ be an $[n,\,k,\,d_{H}(\mathcal{C})]$ constacyclic code over $\mathbb{F}_{q}$ with $2\leq d_{H}(\mathcal{C})<n$.
Then we have $d_{p}(\mathcal{C})\geq d_{H}(\mathcal{C})+2$ if and only if $\mathcal{C}$ is not an MDS code, i.e., $k<n-d_{H}(\mathcal{C})+1$.
\end{lemma}

\begin{lemma}\label{lemdH}
Let $\mathcal{C}=\left\langle g(x)\right\rangle$ be a repeated-root cyclic code of length $lp^{e}$ over $\mathbb{F}_{q}$ and $c(x)=\left(x^{l}-1\right)^{t}v(x)$ a codeword in $\mathcal{C}$ with Hamming weight $d_{H}(\mathcal{C})$, where $l$ and $e$ are positive integers with ${\rm gcd}\left(l,\,p\right)=1$, $0\leq t\leq p^{e}-1$ and $\left(x^{l}-1\right)\nmid v(x)$.
Then
\begin{equation*}
  w_{H}\left(c(x)\right)=P_{t}\cdot N_{v}
\end{equation*}
where $P_{t}$ is defined as $\left(\ref{eqpt}\right)$ in Lemma \ref{lemdistance} and $N_{v}=w_{H}\left(v(x)\,{\rm mod}\left(x^{l}-1\right)\right)$.
\end{lemma}

{\it Proof}\,
Denote $\overline{v}\left(x\right)=\left(v(x)\,{\rm mod}\left(x^{l}-1\right)\right)$ and
\begin{equation*}
  \overline{c}_{t}\left(x\right)=\left(\left(x^{l}-1\right)^{t}\cdot\overline{v}\left(x\right)^{p^e} {\rm mod}\left(x^{lp^{e}}-1\right)\right).
\end{equation*}
Assume that
\begin{equation*}
  g(x)=\prod_{i=1}^{r}m_{i}(x)^{e_{i}}
\end{equation*}
and
\begin{equation*}
  \overline{g}_{t}(x)=\prod_{1\leq i\leq r,\,e_{i}>t}m_{i}(x).
\end{equation*}
It follows from $x^{lp^e}-1=(x^l-1)^{p^e}$, $\left(x^{l}-1\right)\nmid v(x)$ and $g(x)\,|\,c(x)$ that $\overline{g}_{t}(x)\,|\,\overline{v}\left(x\right)$.
Combining with $t<p^e$, one can obtain that for any $1\leq i\leq r$,

$i$) if $e_{i}>t$, then $m_{i}(x)\,|\,\overline{v}\left(x\right)$ and $m_{i}(x)$ is a factor of $\overline{c}_{t}\left(x\right)$ with multiplicity at least $p^e$;

$ii$) if $e_{i}\leq t$, then $m_{i}(x)$ is a factor of $\overline{c}_{t}\left(x\right)$ with multiplicity at least $t$.
\\
Hence $g(x)\,|\,\overline{c}_{t}(x)$.

Meanwhile, due to ${\rm deg}(\overline{v}\left(x\right))<l$, there must exist a root of $x^{l}-1$ whose multiplicity in $\overline{c}_{t}\left(x\right)$ is exactly $t$.
This leads to $(x^{lp^e}-1)\,\nmid\,\overline{c}_{t}\left(x\right)$ and then $\overline{c}_{t}(x)$ is a nonzero codeword in $\mathcal{C}$.
It can be verified that
\begin{equation*}
\begin{split}
  w_{H}\left(\overline{c}_{t}(x)\right)&=w_{H}\left(\left(x^{l}-1\right)^t\cdot \overline{v}\left(x\right)^{p^{e}}{\rm mod}\left(x^{lp^{e}}-1\right)\right)\\
  &\leq w_{H}\left(\left(x^{l}-1\right)^t\cdot \overline{v}\left(x\right)^{p^{e}}\right)
  \leq w_{H}\left(\left(x^{l}-1\right)^{t}\right)\cdot w_{H}\left(\overline{v}\left(x\right)^{p^{e}}\right)
  =P_{t}\cdot N_{v}.
\end{split}
\end{equation*}
On the other hand, according to Theorem $6.3$ in \cite{MCJ}, we have
\begin{equation*}
  w_{H}\left(c(x)\right)\geq w_{H}\left((x^{l}-1)^{t}\right)\cdot w_{H}\left(v(x)\,{\rm mod}\,(x^{l}-1)\right)=P_{t}\cdot N_{v}\geq w_{H}\left(\overline{c}_{t}(x)\right).
\end{equation*}
Since $w_{H}(c(x))=d_{H}(\mathcal{C})$, one can immediately conclude that
\begin{equation*}
  w_{H}\left(c(x)\right)=w_{H}\left(\overline{c}_{t}(x)\right)=P_{t}\cdot N_{v}.
\end{equation*}
This completes the proof.
\hfill $\Box$

The following lemma will be frequently used to prove our results.

\begin{lemma}\label{lemsolution}
Let $p$ be a prime power with $5\,|\,(p-1)$, $\beta$ be a primitive $5$-th root of unity in $\mathbb{F}_{p}$ and $a_i\in \mathbb{F}_{p}^*$ for $1\leq i\leq3$.
Then
\begin{equation}\label{eqbelta}
  \beta^2+3\beta+1\ne0
\end{equation}
and for $(i,\,j)=(2,\,3)$, $(2,\,4)$ or $(3,\,4)$, the solution of the $\mathbb{F}_{p}$-linear system of equations
\begin{equation}\label{eqsolution}
  \left\{
  \begin{array}{l}
    1+a_{1}+a_{2}+a_{3}=0,\\[2mm]
    1+a_{1}\,\beta+a_{2}\,\beta^i+a_{3}\,\beta^j=0,\\[2mm]
    1+a_{1}\,\beta^2+a_{2}\,\beta^{2i}+a_{3}\,\beta^{2j}=0\\
  \end{array}
  \right.
\end{equation}
is given as
\begin{table}[!htb]
  \renewcommand\arraystretch{1.1}
  \setlength{\tabcolsep}{25.2mm}{
  \begin{tabular}{@{}ll@{}}
  \hline Value of ($i$,\,$j$)      &Corresponding solution $\left(a_{1},\,a_{2},\,a_{3}\right)$       \\
  \hline $\left(2,\,3\right)$       &$\left(-\frac{\beta^2+\beta+1}{\beta^2},\,\frac{\beta^2+\beta+1}{\beta^3},\,-\frac{1}{\beta^3}\right)$ \\
  \hline $\left(2,\,4\right)$       &$\left(-\frac{1}{\beta},\,-\frac{\beta}{\beta+1},\,\frac{1}{\beta\left(\beta+1\right)}\right)$ \\
  \hline $\left(3,\,4\right)$       &$\left(\frac{\beta^2}{\beta+1},\,-\frac{1}{\beta+1},\,-\beta\right)$. \\
  \hline
  \end{tabular}}
\end{table}
\end{lemma}

{\it Proof}\,
Assume that $\beta^2+3\beta+1=0$.
The fact $\beta$ is a primitive $5$-th root of unity indicates
\begin{equation*}
  0=\beta^4+\beta^3+\beta^2+\beta+1=-5\left(3\beta+1\right)
\end{equation*}
which yields $\beta^2=-3\beta-1=0$, a contradiction.

If $(i,\,j)=(2,\,3)$, then (\ref{eqsolution}) is transformed into
\begin{equation*}
  \left\{
  \begin{array}{l}
    1+a_{1}+a_{2}+a_{3}=0,\\[2mm]
    1+a_{1}\,\beta+a_{2}\,\beta^2+a_{3}\,\beta^3=0,\\[2mm]
    1+a_{1}\,\beta^2+a_{2}\,\beta^4+a_{3}\,\beta=0.\\
  \end{array}
  \right.
\end{equation*}
This leads to
\begin{equation*}\label{eq5p123}
  a_{1}=-\frac{\beta^2+\beta+1}{\beta^2},\quad a_{2}=\frac{\beta^2+\beta+1}{\beta^3},\quad a_{3}=-\frac{1}{\beta^3}.
\end{equation*}
Similarly, we can derive the solutions of (\ref{eqsolution}) for $(i,\,j)=(2,\,4)$ and (3,\,4).
This completes the proof.
\hfill $\Box$

\section{Constructions of MDS Symbol-Pair Codes}

In this section, we propose three new classes of MDS symbol-pair codes from repeated-root cyclic codes by analyzing the solutions of certain equations over $\mathbb{F}_{p}$.   Firstly, for length $5p$, two classes of MDS symbol-pair codes with minimum symbol-pair distance $7$ or $8$ are constructed respectively.  In addition, for $n=4p$, we derive a class of MDS symbol-pair codes with $d_{p}=7$.

From now on, we denote $c^{(k)}(x)$ by the $k$-th formal derivative of $c(x)$, where $k$ is a positive integer and $c(x)\in \mathbb{F}_{p}[x]$.  Let $\star$ denote an element in $\mathbb{F}_p^{*}$ and $\mathbf{0}$ is the zero vector.  Due to the linearity and the cyclic shift property of cyclic codes, we assume that the constant term of $c(x)$ occurred in this paper is always $1$.

\subsection{MDS symbol-pair codes for $n=5p$}

In this subsection, two classes of MDS symbol-pair codes with length $5p$ are constructed.

Now we present a class of MDS symbol-pair codes with minimum symbol-pair distance $7$ for any prime $p$ with $5\,|\left(p-1\right)$ and $p\ne 41$.

\begin{thm}\label{thmMDS5p7p}
Let $p$ be a prime with $5\,|\left(p-1\right)$ and $p\ne 41$.
Then there exists an MDS $\left(5p,\,7\right)_{p}$ symbol-pair code.
\end{thm}

{\it Proof}\,
Let $\mathcal{C}$ be a repeated-root cyclic code of length $5p$ over $\mathbb{F}_{p}$ with the generator polynomial
\begin{equation*}
  g(x)=\left(x-1\right)^{3}\left(x-\beta\right)\left(x-\beta^{2}\right)
\end{equation*}
where $\beta$ is a primitive $5$-th root of unity in $\mathbb{F}_{p}$.

Note that $\mathcal{C}$ is a $\left[5p,\,5p-5,\,4\right]$ cyclic code due to Lemma \ref{lemdistance}.
Indeed, recall that $\overline{g}_{t}(x)$ is the generator polynomial of $\overline{\mathcal{C}}_{t}$.
If $t=0$, then
\begin{equation*}
  \overline{g}_{0}(x)=\left(x-1\right)\left(x-\beta\right)\left(x-\beta^2\right)
\end{equation*}
and
\begin{equation*}
  P_{0}\cdot d_{H}\left(\overline{\mathcal{C}}_{0}\right)=1\cdot 4=4.
\end{equation*}
If $t=1$, then $\overline{g}_{1}(x)=x-1$ and
\begin{equation*}
  P_{1}\cdot d_{H}\left(\overline{\mathcal{C}}_{1}\right)=2\cdot 2=4.
\end{equation*}
If $t=2$, then $\overline{g}_{2}(x)=x-1$ and
\begin{equation*}
  P_{2}\cdot d_{H}\left(\overline{\mathcal{C}}_{2}\right)=3\cdot 2=6.
\end{equation*}
If $3\leq t\leq p-1$, then $\overline{g}_{t}(x)=1$ and
\begin{equation*}
  P_{t}\cdot d_{H}\left(\overline{\mathcal{C}}_{t}\right)=(t+1)\cdot 1=t+1\geq 4.
\end{equation*}
With the aid of the equality $\left(\ref{eqdistance}\right)$ in Lemma \ref{lemdistance}, one can immediately get $d_{H}(\mathcal{C})=4$.

Since $\mathcal{C}$ is not MDS, by Lemma \ref{lemMDS}, one can obtain that $d_{p}(\mathcal{C})\geq 6$.
Now we claim that there does not exist a codeword in $\mathcal{C}$ with $(w_H(c(x)), \,w_p(c(x)))=(5,\,6)$.
On the contrary, without loss of generality, we assume
\begin{equation*}
  c(x)=c_{0}+c_{1}\,x+c_{2}\,x^{2}+c_{3}\,x^{3}+c_{4}\,x^{4}
\end{equation*}
where $c_{i}\in \mathbb{F}_{p}^{*}$ for any $0\leq i\leq 4$.
This is contradictory with
\begin{equation*}
  {\rm deg}(g(x))=5,\quad {\rm deg}(c(x))\geq {\rm deg}(g(x)).
\end{equation*}
Thus, there does not exist a codeword in $\mathcal{C}$ with $(w_H(c(x)), \,w_p(c(x)))=(5,\,6)$.
To show that $\mathcal{C}$ is an MDS $\left(5p,\,7\right)_{p}$ symbol-pair code, it is sufficient to verify that there does not exist a codeword in $\mathcal{C}$ with $(w_H(c(x)), \,w_p(c(x)))=(4,\,6)$.

Let $c(x)$ be a codeword in $\mathcal{C}$ with Hamming weight $4$.
Suppose that $c(x)$ has the factorization $c(x)=\left(x^{5}-1\right)^{t}v(x)$,
where $0\leq t\leq p-1$, $\left(x^{5}-1\right)\nmid v(x)$ and
\begin{equation*}
  v(x)=v_{0}(x^{5})+x\,v_{1}(x^{5})+x^{2}\,v_{2}(x^{5})+x^{3}\,v_{3}(x^{5})
  +x^{4}\,v_{4}(x^{5}).
\end{equation*}
It follows from Lemma \ref{lemdH} that
\begin{equation*}
  4=w_{H}\left(\left(x^{5}-1\right)^{t}\right)\cdot w_{H}\left(v(x)\,{\rm mod}\left(x^{5}-1\right)\right)=\left(1+t\right)N_{v}
\end{equation*}
where $N_{v}=w_{H}\left(v(x)\,{\rm mod}\left(x^{5}-1\right)\right)$.
Then one can deduce that $\left(N_{v},\,t\right)=\left(1,\,3\right)$, $\left(2,\,1\right)$ or $\left(4,\,0\right)$.

  If $\left(N_{v},\,t\right)=\left(1,\,3\right)$, then it is obvious that the symbol-pair weight of $c(x)$ is greater than $6$.

  If $\left(N_{v},\,t\right)=\left(2,\,1\right)$ and $c(x)$ has symbol-pair weight $6$, then Lemma \ref{lemdpdH} indicates that its certain cyclic shift must have the form
  \begin{equation*}
    \left(\star,\,\star,\,\mathbf{0},\,\star,\,\star,\,\mathbf{0}\right).
  \end{equation*}
  Let
  \begin{equation*}
    c(x)=1+a_{1}\,x+a_{2}\,x^{5i}+a_{3}\,x^{5i+1}
  \end{equation*}
  for some positive integer $i$ with $1\leq i\leq p-1$ and $a_{1},\,a_{2},\,a_{3}\in \mathbb{F}_{p}^{*}$.
  It follows from $5\,|\left(p-1\right)$ and ${\rm gcd}\left(i,p\right)=1$ that $p\nmid 5i$.
  The fact $c\left(1\right)=c\left(\beta\right)=0$ induces that
  \begin{equation*}
    \left\{
    \begin{array}{l}
    1+a_{1}+a_{2}+a_{3}=0,\\[2mm]
    1+a_{1}\,\beta+a_{2}+a_{3}\,\beta=0\\
    \end{array}
    \right.
  \end{equation*}
  which implies $a_{1}=-a_{3}$ and $a_{2}=-1$.
  Then $c^{\left(1\right)}\left(1\right)=c^{\left(2\right)}\left(1\right)=0$ yields
  \begin{equation*}
    \left\{
    \begin{array}{l}
    a_{1}-5i-\left(5i+1\right)a_{1}=0,\\[2mm]
    -5i\left(5i-1\right)-5i\left(5i+1\right)a_{1}=0.\\
    \end{array}
    \right.
  \end{equation*}
  This indicates $a_1=-1$ and then $2=0$, a contradiction.

  If $\left(N_{v},\,t\right)=\left(4,\,0\right)$ and $c(x)$ has symbol-pair weight $6$, then its corresponding cyclic shift must have the form
  \begin{equation*}
    \left(\star,\,\star,\,\mathbf{0},\,\star,\,\star,\,\mathbf{0}\right)
  \end{equation*}
  or
  \begin{equation*}
    \left(\star,\,\star,\,\star,\,\mathbf{0},\,\star,\,\mathbf{0}\right).
  \end{equation*}
  In what follows, we discuss the above two cases one by one.

{\bf Case I}\,\,  For the case $\left(\star,\,\star,\,\mathbf{0},\,\star,\,\star,\,\mathbf{0}\right)$, there are two subcases to be considered:
\begin{itemize}
  \item
  For the subcase $c(x)=1+a_1\,x+a_2\,x^{5i+2}+a_3\,x^{5i+3}$ with $1\leq i\leq p-1$ and $a_1,\,a_2,\,a_3\in \mathbb{F}_p^*$, it follows from $c\left(1\right)=c^{\left(1\right)}\left(1\right)=c^{\left(2\right)}\left(1\right)=0$ that
  \begin{equation}\label{eqMDS5p7pNv4220101}
    \left\{
    \begin{array}{l}
     1+a_{1}+a_{2}+a_{3}=0,\\[2mm]
     a_{1}+\left(5i+2\right)a_{2}+\left(5i+3\right)a_{3}=0,\\[2mm]
     \left(5i+2\right)\left(5i+1\right)a_{2}
     +\left(5i+3\right)\left(5i+2\right)a_{3}=0.\\
    \end{array}
    \right.
  \end{equation}

  If $p\,|\left(5i+2\right)$, then $\left(\ref{eqMDS5p7pNv4220101}\right)$ implies that $a_1=-a_3$ and $a_2=-1$.
  Then $c\left(\beta\right)=c\left(\beta^2\right)=0$ yields
  \begin{equation*}
    \left\{
    \begin{array}{l}
      1+a_{1}\,\beta-\beta^2-a_{1}\,\beta^3=0,\\[2mm]
      1+a_{1}\,\beta^2-\beta^4-a_{1}\,\beta^6=0.\\
    \end{array}
    \right.
  \end{equation*}
  One can immediately obtain that
  \begin{equation*}
    a_1=\frac{\beta^2-1}{\beta-\beta^3}=\frac{\beta^4-1}{\beta^2-\beta^6}.
  \end{equation*}
  This leads to $\beta=1$, a contradiction.

  If $p\nmid\left(5i+2\right)$, then $\left(\ref{eqMDS5p7pNv4220101}\right)$ yields that $a_1=-a_2$ and $a_3=-1$.
  It follows from $c\left(\beta\right)=c\left(\beta^2\right)=0$ that
  \begin{equation*}
  \left\{
  \begin{array}{l}
    1+a_{1}\,\beta-a_{1}\,\beta^2-\beta^3=0,\\[2mm]
    1+a_{1}\,\beta^2-a_{1}\,\beta^4-\beta^6=0.\\
  \end{array}
  \right.
  \end{equation*}
  Then one gets that
  \begin{equation*}
    a_1=\frac{\beta^3-1}{\beta-\beta^2}=\frac{\beta^6-1}{\beta^2-\beta^4}
  \end{equation*}
  which induces
  \begin{equation*}
    \beta^3+1=\beta\left(\beta+1\right).
  \end{equation*}
  This implies $\left(\beta-1\right)\left(\beta^2-1\right)=0$, a contradiction.

  \item
  Consider the subcase $c(x)=1+a_1\,x+a_2\,x^{5i+3}+a_3\,x^{5i+4}$ with $0\leq i\leq p-2$ and $a_1,\,a_2,\,a_3\in \mathbb{F}_p^*$.
  By arguments similar to the previous subcase of $c(x)=1+a_1\,x+a_2\,x^{5i+2}+a_3\,x^{5i+3}$, one can also derive a contradiction and we omit the proof here.
\end{itemize}

  {\bf Case II}\,\,  For the remaining case  $\left(\star,\,\star,\,\star,\,\mathbf{0},\,\star,\,\mathbf{0}\right)$, there are also two subcases to be discussed:

\begin{itemize}
  \item
  Consider the subcase $c(x)=1+a_{1}\,x+a_{2}\,x^{2}+a_{3}\,x^{5i+3}$ with $1\leq i\leq p-1$ and $a_{1},\,a_{2},\,a_{3}\in \mathbb{F}_p^{*}$.
  Notice that $c\left(1\right)=c\left(\beta\right)=c\left(\beta^2\right)=0$ and Lemma \ref{lemsolution} indicates
  \begin{equation}\label{eqMDS5p7pNv4310101}
    a_{1}=-\frac{\beta^2+\beta+1}{\beta^2},\quad a_{2}=\frac{\beta^2+\beta+1}{\beta^3},\quad a_{3}=-\frac{1}{\beta^3}.
  \end{equation}
  It follows from $c^{\left(1\right)}\left(1\right)=c^{\left(2\right)}\left(1\right)=0$ that
  \begin{equation}\label{eqMDS5p7pNv4310102}
  \left\{
  \begin{array}{l}
    a_{1}+2\,a_{2}+\left(5i+3\right)a_{3}=0,\\[2mm]
    2\,a_{2}+\left(5i+3\right)\left(5i+2\right)a_{3}=0.\\
  \end{array}
  \right.
  \end{equation}
  Observe that (\ref{eqMDS5p7pNv4310102}) yields
  \begin{equation}\label{eqMDS5p7pNv4310103}
  \left\{
  \begin{array}{l}
    a_{1}=\left(5i+3\right)\left(5i+1\right)a_{3},\\[2mm]
    \left(5i+2\right)a_{1}+2\left(5i+1\right)a_{2}=0\\
  \end{array}
  \right.
  \end{equation}
  and the second equality in (\ref{eqMDS5p7pNv4310102}) indicates $p\nmid (5i+2)$.
  Let $t=5i+2$.
  By $\left(\ref{eqMDS5p7pNv4310101}\right)$ and (\ref{eqMDS5p7pNv4310103}), one can immediately have
  \begin{equation}\label{eqMDS5p7pNv4310104}
  \left\{
  \begin{array}{l}
    t^2=\beta^3+\beta^2+\beta+1,\\[2mm]
    t(\beta-2)=2.\\
  \end{array}
  \right.
  \end{equation}
  The second equality in (\ref{eqMDS5p7pNv4310104}) indicates $\beta\ne2$ and $t=-\frac{2}{\beta-2}$.
  By substituting the value of $t$ into the first equality in (\ref{eqMDS5p7pNv4310104}), one can obtain
  \begin{equation*}
    \frac{4\beta}{(\beta-2)^2}=(\beta^3+\beta^2+\beta+1)\beta.
  \end{equation*}
  It follows from $\beta^4+\beta^3+\beta^2+\beta+1=0$ that $\beta^2=-4$ and
  \begin{equation*}
    \beta^4+\beta^3+\beta^2+\beta+1=13-3\beta=0.
  \end{equation*}
  This leads to $\beta=\frac{13}{3}$ and then
  \begin{equation*}
    \beta^2=\frac{169}{9}=-4
  \end{equation*}
  implies $5\cdot 41=0$, which is contradictory with $5\,|\,(p-1)$ and $p\ne 41$.

  \item
  For the subcase $c(x)=1+a_{1}\,x+a_{2}\,x^{2}+a_{3}\,x^{5i+4}$ with $0\leq i\leq p-2$ and $a_{1},\,a_{2},\,a_{3}\in \mathbb{F}_p^{*}$, it follows from $c\left(1\right)=c\left(\beta\right)=c\left(\beta^2\right)=0$ and Lemma \ref{lemsolution} that
  \begin{equation}\label{eqMDS5p7pNv4310201}
    a_{1}=-\frac{1}{\beta},\quad a_{2}=-\frac{\beta}{\beta+1},\quad a_{3}=\frac{1}{\beta\left(\beta+1\right)}.
  \end{equation}
  On the other hand, $c^{\left(1\right)}\left(1\right)=c^{\left(2\right)}\left(1\right)=0$ yields that
  \begin{equation*}
  \left\{
  \begin{array}{l}
    a_{1}+2\,a_{2}+\left(5i+4\right)a_{3}=0,\\[2mm]
    2\,a_{2}+\left(5i+4\right)\left(5i+3\right)a_{3}=0\\
  \end{array}
  \right.
  \end{equation*}
  which is equivalent to
  \begin{equation*}
  \left\{
  \begin{array}{l}
    a_{1}+2\,a_{2}+\left(5i+4\right)a_{3}=0,\\[2mm]
    a_{1}=\left(5i+4\right)\left(5i+2\right)a_{3}.\\
  \end{array}
  \right.
  \end{equation*}
  Let $t=5i+3$.
  Together with $\left(\ref{eqMDS5p7pNv4310201}\right)$, one can immediately obtain that
  \begin{equation*}
  \left\{
  \begin{array}{l}
    t=2\,\beta^2+\beta,\\[2mm]
    t^2+\beta=0.\\
  \end{array}
  \right.
  \end{equation*}
  Then by substituting the value of $t$, one has
  \begin{equation}\label{eqMDS5p7pNv4310202}
    \beta^2=3\,\beta+3
  \end{equation}
  and
  \begin{equation*}
    0=\beta^4+\beta^3+\beta^2+\beta+1=61\,\beta+49.
  \end{equation*}
  If $p=61$, then $49=0$, which is impossible.
  If $p\ne 61$, then $\beta=-\frac{49}{61}$.
  It follows from (\ref{eqMDS5p7pNv4310202}) that
  \begin{equation*}
    \left(\frac{49}{61}\right)^2=3\left(-\frac{49}{61}+1\right).
  \end{equation*}
  This implies $5\cdot 41=0$, a contradiction similar as the previous subcase.
\end{itemize}

Therefore, $\mathcal{C}$ is an MDS $\left(5p,\,7\right)_{p}$ symbol-pair code.
This completes the proof.
\hfill $\Box$

Another class of MDS symbol-pair codes with $n=5p$ and $d_{p}=8$ is proposed as follows.

\begin{thm}\label{thmMDS5p8p}
Let $p$ be a prime with $5\,|\left(p-1\right)$.
Then there exists an MDS $\left(5p,\,8\right)_{p}$ symbol-pair code.
\end{thm}

{\it Proof}\,
Let $\mathcal{C}$ be a repeated-root cyclic code of length $5p$ over $\mathbb{F}_{p}$ with the generator polynomial
\begin{equation*}
  g(x)=\left(x-1\right)^{3}\left(x-\beta\right)\left(x-\beta^{2}\right)^2
\end{equation*}
where $\beta$ is a primitive $5$-th root of unity in $\mathbb{F}_{p}$.
It can be verified that $\mathcal{C}$ is an MDS $\left(5p,\,8\right)_{p}$ symbol-pair code by similar techniques used in the proof of Theorem \ref{thmMDS5p7p}.
Since the proof is lengthy and some of them seems a bit cumbersome, we present it in the $\bf{Appendix}$.
\hfill $\Box$

Now we provide two examples to illustrate the constructions in Theorems \ref{thmMDS5p7p} and \ref{thmMDS5p8p}.

\begin{Example}
(1) Let $\mathcal{C}$ be a repeated-root cyclic code of length $55$ over $\mathbb{F}_{11}$ with the generator polynomial
\begin{equation*}
  g(x)=\left(x-1\right)^{3}\left(x-3\right)\left(x-3^2\right).
\end{equation*}
MAGMA experiments show that $\mathcal{C}$ is a $[55,\,50,\,4]$ code and the minimum symbol-pair distance of $\mathcal{C}$ is $7$, which satisfies our result in Theorem \ref{thmMDS5p7p}.

(2) Let $\mathcal{C}$ be a repeated-root cyclic code of length $55$ over $\mathbb{F}_{11}$ with the generator polynomial
\begin{equation*}
  g(x)=\left(x-1\right)^{3}\left(x-3\right)\left(x-3^2\right)^2.
\end{equation*}
By a MAGMA program, it can be checked that $\mathcal{C}$ is a $[55,\,49,\,4]$ code and the minimum symbol-pair distance of $\mathcal{C}$ is $8$, which is consistent with our result in Theorem \ref{thmMDS5p8p}.
\end{Example}

\subsection{MDS symbol-pair codes for $n=4p$}

In this subsection, we shall construct a class of MDS symbol-pair codes with $d_{p}=7$, which generalizes Theorem $3.8$ in \cite{KZZLC}.

\begin{theorem}\label{thmMDS4p7p}
Let $p$ be an odd prime.
Then there exists an MDS $\left(4p,\,7\right)_{p}$ symbol-pair code.
\end{theorem}

{\it Proof}\,
The case $p\equiv 3\left({\rm mod}\,4\right)$ has been settled, see Theorem $3.8$ in \cite{KZZLC}.
For the case $p\equiv 1\left({\rm mod}\,4\right)$,
let $\mathcal{C}$ be a repeated-root cyclic code of length $4p$ over $\mathbb{F}_p$ with the generator polynomial
\begin{equation*}
  g(x)=\left(x-1\right)^{3}\left(x-\omega\right)\left(x+\omega\right)
\end{equation*}
where $\omega$ is a primitive $4$-th root of unity in $\mathbb{F}_{p}$.
In the following, we will claim that for $p\equiv 1\left({\rm mod}\,4\right)$, the code $\mathcal{C}$ is also an MDS $\left(4p,\,7\right)_{p}$ symbol-pair code.

By Lemma \ref{lemdistance}, one can derive that the parameter of $\mathcal{C}$ is $\left[4p,\,4p-5,\,4\right]$.
Since $\mathcal{C}$ is not MDS, by Lemma \ref{lemMDS}, we get $d_{p}(\mathcal{C})\geq 6$.
With a similar argument as Theorem \ref{thmMDS5p7p}, one can obtain that there does not exist a codeword $c(x)$ in $\mathcal{C}$ with $(w_H(c(x)), \,w_p(c(x)))=(5,\,6)$.
In order to show that $\mathcal{C}$ is an MDS $\left(4p,\,7\right)_{p}$ symbol-pair code, we need to prove that there does not exist a codeword in $\mathcal{C}$ with $(w_H(c(x)),\,w_p(c(x)))=(4,\,6)$.

Let $c(x)$ be a codeword in $\mathcal{C}$ with $(w_H(c(x)),\,w_p(c(x)))=(4,\,6)$.
Then Lemma \ref{lemdpdH} indicates that its certain cyclic shift must have the form
\begin{equation*}
  \left(\star,\,\star,\,\star,\,\mathbf{0},\,\star,\,\mathbf{0}\right)
\end{equation*}
or
\begin{equation*}
  \left(\star,\,\star,\,\mathbf{0},\,\star,\,\star,\,\mathbf{0}\right).
\end{equation*}

For the case $\left(\star,\,\star,\,\star,\,\mathbf{0},\,\star,\,\mathbf{0}\right)$, we assume that
\begin{equation*}
  c(x)=1+a_{1}\,x+a_{2}\,x^{2}+a_{3}\,x^{l}
\end{equation*}
for some positive integer $l$ with $4\leq l\leq 4p-2$ and $a_{1},\,a_{2},\,a_{3}\in \mathbb{F}_{p}^{*}$.
It follows from $c\left(1\right)=c^{\left(1\right)}\left(1\right)=c^{(2)}\left(1\right)=0$
that
\begin{equation*}
  \left\{
  \begin{array}{l}
   1+a_{1}+a_{2}+a_{3}=0,\\[2mm]
   a_{1}+2\,a_{2}+l\,a_{3}=0,\\[2mm]
   2\,a_{2}+l\left(l-1\right)a_{3}=0.\\
  \end{array}
  \right.
\end{equation*}
This yields
\begin{equation}\label{eqMDS4p7p3101}
  a_{1}=-\frac{2l}{l-1},\quad a_{2}=\frac{l}{l-2},
  \quad a_{3}=-\frac{2}{\left(l-1\right)\left(l-2\right)}.
\end{equation}

\begin{itemize}
  \item
  If $l$ is even, then we have
  \begin{equation*}
  \left\{
  \begin{array}{l}
    1+a_{1}\,\omega-a_{2}+a_{3}\,\omega^{l}=0,\\[2mm]
    1-a_{1}\,\omega-a_{2}+a_{3}\,\omega^{l}=0\\
  \end{array}
  \right.
  \end{equation*}
  since $c\left(\omega\right)=c\left(-\omega\right)=0$ and $\omega^{2}=-1$.
  It follows that $a_{1}=0$, which is impossible.

  \item
  If $l$ is odd, then $c\left(\omega\right)=c\left(-\omega\right)=0$ indicates that
  \begin{equation*}
  \left\{
  \begin{array}{l}
    1+a_{1}\,\omega-a_{2}+a_{3}\,\omega^{l}=0,\\[2mm]
    1-a_{1}\,\omega-a_{2}-a_{3}\,\omega^{l}=0.\\
  \end{array}
  \right.
  \end{equation*}
  This implies that $a_{2}=1$, which contradicts with the result in $\left(\ref{eqMDS4p7p3101}\right)$.
\end{itemize}

For the remaining case $\left(\star,\,\star,\,\mathbf{0},\,\star,\,\star,\,\mathbf{0}\right)$, we suppose that
\begin{equation*}
  c(x)=1+a_{1}\,x+a_{2}\,x^{l}+a_{3}\,x^{l+1}
\end{equation*}
for some positive integer $l$ with $3\leq l\leq 4p-3$ and $a_{1},\,a_{2},\,a_{3}\in \mathbb{F}_{p}^{*}$.
Then $c\left(1\right)=c^{\left(1\right)}\left(1\right)=c^{\left(2\right)}\left(1\right)=0$ indicates that
\begin{equation}\label{eqMDS4p7p2201}
  \left\{
  \begin{array}{l}
   1+a_{1}+a_{2}+a_{3}=0,\\[2mm]
   a_{1}+l\,a_{2}+\left(l+1\right)a_{3}=0,\\[2mm]
   l\left(l-1\right)a_{2}+l\left(l+1\right)a_{3}=0.\\
  \end{array}
  \right.
\end{equation}
It follows from $c\left(\omega\right)=c\left(-\omega\right)=0$ that
\begin{equation}\label{eqMDS4p7p2202}
  \left\{
  \begin{array}{l}
   1+a_{1}\,\omega+a_{2}\,\omega^{l}+a_{3}\,\omega^{l+1}=0,\\[2mm]
   1-a_{1}\,\omega+a_{2}\left(-\omega\right)^{l}+a_{3}\left(-\omega\right)^{l+1}=0.\\
  \end{array}
  \right.
\end{equation}
Now we divide the proof into the following two subcases:

\begin{itemize}
  \item
  For the subcase $p\,|\,l$, $\left(\ref{eqMDS4p7p2201}\right)$ yields that
  \begin{equation}\label{eqMDS4p7p220101}
    a_{1}+a_{3}=0,\quad a_{2}=-1.
  \end{equation}
  If $l$ is even, then we have $l=2p$ due to $3\leq l\leq 4p-3$.
  It follows from $\left(\ref{eqMDS4p7p2202}\right)$ and $\left(\ref{eqMDS4p7p220101}\right)$ that
  \begin{equation*}
    1=\omega^{l}=\omega^{2p}=\left(-1\right)^p
  \end{equation*}
  which is impossible.
  Similarly, if $l$ is odd, one can obtain that $\omega^{2l}=1$, a contradiction.

  \item
  For the subcase $p\nmid l$, it follows from $\left(\ref{eqMDS4p7p2201}\right)$ that
  \begin{equation}\label{eqMDS4p7p220201}
    a_{1}=-\frac{l+1}{l-1},\quad a_{2}=\frac{l+1}{l-1},\quad a_{3}=-1.
  \end{equation}
  If $l$ is even, then by $\left(\ref{eqMDS4p7p2202}\right)$ and $\left(\ref{eqMDS4p7p220201}\right)$, one can deduce that
  \begin{equation*}
    a_1=\omega^l, \quad 1+a_2\,\omega^l=0.
  \end{equation*}
  Then one can obtain that
  \begin{equation*}
    1=a_{1}^2=\left(-\frac{l+1}{l-1}\right)^{2}.
  \end{equation*}
  This implies $4l=0$, a contradiction.
  By a similar manner, for odd $l$, one can derive that $\omega^{l+1}=\omega^{l-1}=1$, which is impossible.
\end{itemize}

Consequently, $\mathcal{C}$ is an MDS $\left(4p,\,7\right)_{p}$ symbol-pair code.
This proves the desired conclusion.
\hfill $\Box$

Now we give an example to illustrate the construction in Theorem \ref{thmMDS4p7p}.

\begin{Example}
Let $\mathcal{C}$ be a repeated-root cyclic code of length $20$ over $\mathbb{F}_{5}$ with the generator polynomial
\begin{equation*}
  g(x)=\left(x-1\right)^{3}\left(x-2\right)\left(x+2\right).
\end{equation*}
It can be checked by MAGMA that $\mathcal{C}$ is a $[20,\,15,\,4]$ code and the minimum symbol-pair distance of $\mathcal{C}$ is $7$, which coincides with our result in Theorem \ref{thmMDS4p7p}.
\end{Example}

\section{Conclusions and future work}

In this paper, three new classes of MDS symbol-pair codes over $\mathbb{F}_{p}$ with $p$ an odd prime were constructed from repeated-root cyclic codes.
Firstly, for $n=5p$, two classes of MDS symbol-pair codes with minimum symbol-pair distance seven or eight were presented.
In addition, for length $n=4p$, we derived a class of MDS symbol-pair codes with $d_{p}=7$ and our construction extends the result in \cite{KZZLC}.
Note that by utilizing repeated-root cyclic codes, one can construct MDS symbol-pair codes by transforming the problem into analyzing the solutions of certain equations over finite fields.

However, it seems impracticable to construct $(5q,\,7)_p$, $(5q,\,8)_p$ and $(4q,\,7)_p$ MDS symbol-pair codes with $q$ being a power of $p$ using the techniques in Theorems $1$-$3$. For instance, for the case $q=p^2$, $5\,|\,(q-1)$, let $\mathcal{C}$ be a repeated-root cyclic code of length $5q$ over $\mathbb{F}_q$ with the generator polynomial of the form
\begin{equation*}
  g(x)=\left(x-1\right)^{e_1}\left(x-\omega\right)^{e_2}\left(x-\omega^{2}\right)^{e_3}
  \left(x-\omega^{3}\right)^{e_4}\left(x-\omega^{4}\right)^{e_5}
\end{equation*}
where $\omega$ is a primitive $5$-th root of unity in $\mathbb{F}_{q}$.
It can be checked that $\mathcal{C}$ is not an MDS symbol-pair code. It needs further study to construct MDS symbol-pair codes with larger minimum symbol-pair distance and length $lq$, where $q=p^m$ with $m>1$.

\hspace*{\fill}\\
\noindent
{\bf Appendix Proof of Theorem \ref{thmMDS5p8p}:}
\hspace*{\fill}\\
\\
Recall that $\mathcal{C}$ is a repeated-root cyclic code of length $5p$ over $\mathbb{F}_{p}$ with the generator polynomial
\begin{equation*}
  g(x)=\left(x-1\right)^{3}\left(x-\beta\right)\left(x-\beta^{2}\right)^{2}
\end{equation*}
where $\beta$ is a primitive $5$-th root of unity in $\mathbb{F}_{p}$.
By Lemma \ref{lemdistance}, one can derive that the parameter of $\mathcal{C}$ is $\left[5p,\,5p-6,\,4\right]$.
Since $\mathcal{C}$ is not MDS, Lemma \ref{lemMDS} yields that $d_{p}(\mathcal{C})\geq 6$.
Similar as Theorem \ref{thmMDS5p7p}, one can derive that there does not exist a codeword $c(x)$ in $\mathcal{C}$ with $(w_H(c(x)), \,w_p(c(x)))=(5,\,6)$ or $(6,\,7)$.
To prove that $\mathcal{C}$ is an MDS $\left(5p,\,8\right)_{p}$ symbol-pair code, it suffices to determine that there does not exist a codeword in $\mathcal{C}$ with $(w_H(c(x)), \,w_p(c(x)))=(4,\,6)$, $(4,\,7)$ or $(5,\,7)$.

{\bf Case I}:\,\,  $(w_H(c(x)), \,w_p(c(x)))=(4,\,6)$.
Since $\mathcal{C}$ is the subcode of the code in Theorem \ref{thmMDS5p7p} and the proof of Theorem \ref{thmMDS5p7p} indicates that there does not exist a codeword $c(x)$ in $\mathcal{C}$ with $(w_H(c(x)), \,w_p(c(x)))=(4,\,6)$ unless $p=41$.
Now it is sufficient to show that for $p=41$, there does not exist a codeword $c(x)$ in $\mathcal{C}$ with $(w_H(c(x)), \,w_p(c(x)))=(4,\,6)$.
More precisely, we just need to consider {\bf{Case II}} in Theorem \ref{thmMDS5p7p}.
There are two subcases to be discussed:

\begin{itemize}
  \item
  Consider the subcase $c(x)=1+a_{1}\,x+a_{2}\,x^{2}+a_{3}\,x^{5i+3}$ with $1\leq i\leq p-1$ and $a_{1},\,a_{2},\,a_{3}\in \mathbb{F}_p^{*}$.
  Notice that $c\left(1\right)=c\left(\beta\right)=c\left(\beta^2\right)=0$ and Lemma \ref{lemsolution} induces
  \begin{equation}\label{eqMDS5p8pNv4310101}
    a_{1}=-\frac{\beta^2+\beta+1}{\beta^2},\quad a_{2}=\frac{\beta^2+\beta+1}{\beta^3},\quad a_{3}=-\frac{1}{\beta^3}.
  \end{equation}
  It follows from $c^{\left(1\right)}\left(1\right)=c^{\left(1\right)}\left(\beta^2\right)=0$ that
  \begin{equation*}
  \left\{
  \begin{array}{l}
    a_{1}+2\,a_{2}+\left(5i+3\right)a_{3}=0,\\[2mm]
    a_{1}+2\,a_{2}\,\beta^2+\left(5i+3\right)a_{3}\,\beta^4=0\\
  \end{array}
  \right.
  \end{equation*}
  which yields
  \begin{equation*}
    a_1\left(\beta^4-1\right)+2\,a_2\left(\beta^4-\beta^2\right)=0.
  \end{equation*}
  Combining with (\ref{eqMDS5p8pNv4310101}), one can get $\left(\beta-1\right)^2=0$, a contradiction.

  \item
  For the subcase $c(x)=1+a_{1}\,x+a_{2}\,x^{2}+a_{3}\,x^{5i+4}$ with $0\leq i\leq p-2$ and $a_{1},\,a_{2},\,a_{3}\in \mathbb{F}_p^{*}$, by $c\left(1\right)=c\left(\beta\right)=c\left(\beta^2\right)=0$ and Lemma \ref{lemsolution}, one can obtain that
  \begin{equation}\label{eqMDS5p8pNv4310201}
    a_{1}=-\frac{1}{\beta},\quad a_{2}=-\frac{\beta}{\beta+1},\quad a_{3}=\frac{1}{\beta\left(\beta+1\right)}.
  \end{equation}
  On the other hand,
  $c^{\left(1\right)}\left(1\right)=c^{\left(1\right)}\left(\beta^2\right)=0$ that
  \begin{equation*}
  \left\{
  \begin{array}{l}
    a_{1}+2\,a_{2}+\left(5i+4\right)a_{3}=0,\\[2mm]
    a_{1}+2\,a_{2}\,\beta^2+\left(5i+4\right)a_{3}\,\beta=0\\
  \end{array}
  \right.
  \end{equation*}
  which induces
  \begin{equation*}
    a_1\left(\beta-1\right)=2\,a_2\left(\beta^2-\beta\right).
  \end{equation*}
  Together with $\left(\ref{eqMDS5p8pNv4310201}\right)$, one can immediately obtain that
  \begin{equation*}
    2\,\beta^3=\beta+1.
  \end{equation*}
  This leads to
  \begin{equation*}
    \left(\beta-1\right)\left(2\,\beta^2+2\,\beta+1\right)=0.
  \end{equation*}
  The fact $\beta$ is a primitive $5$-th root of unity implies that $2\beta^2+2\beta+1=0$ and then one has
  \begin{equation*}
    \beta^2+\beta=-\left(\beta^2+\beta+1\right)=\beta^4+\beta^3
  \end{equation*}
  which is impossible.
\end{itemize}

{\bf Case II}:\,\,  $(w_H(c(x)), \,w_p(c(x)))=(4,\,7)$.
For this case, Lemma \ref{lemdpdH} implies that the cyclic shift of $c(x)$ must have the form
\begin{equation*}
  \left(\star,\,\star,\,\mathbf{0},\,\star,\,\mathbf{0},\,\star,\,\mathbf{0}\right).
\end{equation*}
Assume that $c(x)=\left(x^{5}-1\right)^{t}v(x)$, where $0\leq t\leq p-1$, $\left(x^{5}-1\right)\nmid v(x)$ and
\begin{equation*}
  v(x)=v_{0}(x^{5})+x\,v_{1}(x^{5})+x^{2}\,v_{2}(x^{5})+x^{3}\,v_{3}(x^{5})+x^{4}\,v_{4}(x^{5}).
\end{equation*}
Recall that $N_{v}=w_{H}\left(v(x)\,{\rm mod}\left(x^{5}-1\right)\right)$.
Then by Lemma \ref{lemdH}, one can deduce that
\begin{equation*}
  4=w_{H}\left(\left(x^{5}-1\right)^{t}\right)\cdot w_{H}\left(v(x)\,{\rm mod}\left(x^{5}-1\right)\right)=\left(1+t\right)N_{v}.
\end{equation*}

If $\left(N_{v},\,t\right)=\left(1,\,3\right)$, then it is easily seen that the symbol-pair weight of $c(x)$ is greater than $7$.

If $\left(N_{v},\,t\right)=\left(2,\,1\right)$, then there are three subcases to be discussed:

  (1) For the subcase $c(x)=1+a_{1}\,x+a_{2}\,x^{5i}+a_{3}\,x^{5j}$ with $1\leq i<j \leq p-1$ and $a_{1},\,a_{2},\,a_{3}\in \mathbb{F}_p^{*}$, it can be verified that
  \begin{equation*}
  \left\{
  \begin{array}{l}
    1+a_{1}\,\beta+a_{2}+a_{3}=0,\\[2mm]
    1+a_{1}\,\beta^2+a_{2}+a_{3}=0\\
  \end{array}
  \right.
  \end{equation*}
  since $c\left(\beta\right)=c\left(\beta^2\right)=0$.
  Then one can obtain that $a_{1}=0$, a contradiction.

  (2) For the subcase $c(x)=1+a_{1}\,x+a_{2}\,x^{5i+1}+a_{3}\,x^{5j+1}$ with $1\leq i<j \leq p-1$ and $a_{1},\,a_{2},\,a_{3}\in \mathbb{F}_p^{*}$, by $c\left(1\right)=c\left(\beta\right)=0$, one can get
  \begin{equation*}
  \left\{
  \begin{array}{l}
    1+a_{1}+a_{2}+a_{3}=0,\\[2mm]
    1+a_{1}\,\beta+a_{2}\,\beta+a_{3}\,\beta=0.\\
  \end{array}
  \right.
  \end{equation*}
  This implies that $\beta=1$, which is impossible.

  (3) For the subcase $c(x)=1+a_{1}\,x+a_{2}\,x^{5i}+a_{3}\,x^{5j+1}$ with $1\leq i<j \leq p-1$ and $a_{1},\,a_{2},\,a_{3}\in \mathbb{F}_p^{*}$, it follows from $c^{\left(1\right)}\left(1\right)=c^{\left(1\right)}\left(\beta^2\right)=0$ that
  \begin{equation*}
  \left\{
  \begin{array}{l}
    a_{1}+5i\,a_{2}+\left(5j+1\right)a_{3}=0,\\[2mm]
    a_{1}+5i\,a_{2}\,\beta^3+\left(5j+1\right)a_{3}=0.\\
  \end{array}
  \right.
  \end{equation*}
  This leads to $\beta^3=1$, a contradiction.

  If $\left(N_{v},\,t\right)=\left(4,\,0\right)$, then there are also three subcases to be considered:

  (1) For the subcase $c(x)=1+a_{1}\,x+a_{2}\,x^{5i+2}+a_{3}\,x^{5j+3}$ with $1\leq i<j\leq p-1$ and $a_{1},\,a_{2},\,a_{3}\in \mathbb{F}_p^{*}$, by Lemma \ref{lemsolution} and $c\left(1\right)=c\left(\beta\right)=c\left(\beta^2\right)=0$, one can derive that
  \begin{equation}\label{eq4740101}
    a_{1}=-\frac{\beta^2+\beta+1}{\beta^2},\quad a_{2}=\frac{\beta^2+\beta+1}{\beta^3},\quad
    a_{3}=-\frac{1}{\beta^3}.
  \end{equation}
  It follows from $c^{\left(1\right)}\left(1\right)=c^{\left(1\right)}\left(\beta^2\right)=0$ that
  \begin{equation*}
  \left\{
  \begin{array}{l}
    a_{1}+\left(5i+2\right)a_{2}+\left(5j+3\right)a_{3}=0,\\[2mm]
    a_{1}+\left(5i+2\right)a_{2}\,\beta^2+\left(5j+3\right)a_{3}\,\beta^4=0\\
  \end{array}
  \right.
  \end{equation*}
  which indicates
  \begin{equation*}
  \left\{
  \begin{array}{l}
    \left(\beta^4-1\right)a_{1}+(\beta^4-\beta^2)\left(5i+2\right)a_{2}=0,\\[2mm]
    \left(\beta^2-1\right)\left(5i+2\right)a_{2}+\left(5j+3\right)\left(\beta^4-1\right)a_{3}=0.\\
  \end{array}
  \right.
  \end{equation*}
  Together with $\left(\ref{eq4740101}\right)$, one can immediately obtain that
  \begin{equation}\label{eq4740102}
  \left\{
  \begin{array}{l}
    \beta^2+1=\left(5i+2\right)\beta,\\[2mm]
    \left(5i+2\right)\left(\beta^2+\beta+1\right)=\left(5j+3\right)\left(\beta^2+1\right).\\
  \end{array}
  \right.
  \end{equation}
  By substituting the value of $\beta^2+1$ in the first equality into the second equality of (\ref{eq4740102}), we can get
  \begin{equation*}
    \left(5i+2\right)\left(5i+3\right)\beta=\left(5i+2\right)\left(5j+3\right)\beta
  \end{equation*}
  which yields $i=j$ due to $p\nmid (5i+2)$.
  This contradicts with $i<j$.

  (2) Consider the subcase $c(x)=1+a_{1}\,x+a_{2}\,x^{5i+2}+a_{3}\,x^{5j+4}$ with $1\leq i\leq j\leq p-2$ and $a_{1},\,a_{2},\,a_{3}\in \mathbb{F}_p^{*}$.
  The fact $c\left(1\right)=c\left(\beta\right)=c\left(\beta^2\right)=0$ and Lemma \ref{lemsolution} leads to
  \begin{equation}\label{eq474201}
    a_{1}=-\frac{1}{\beta},\quad a_{2}=-\frac{\beta}{\beta+1},\quad a_{3}=\frac{1}{\beta\left(\beta+1\right)}.
  \end{equation}
  It follows from $c^{\left(1\right)}\left(1\right)=c^{\left(1\right)}\left(\beta^2\right)=0$ that
  \begin{equation*}
  \left\{
  \begin{array}{l}
    a_{1}=\beta\left(5i+2\right)a_2,\\[2mm]
    \left(5i+2\right)\left(\beta+1\right)a_{2}+\left(5j+4\right)a_{3}=0.\\
  \end{array}
  \right.
  \end{equation*}
  By substituting $\left(\ref{eq474201}\right)$, one can immediately derive that
  \begin{equation}\label{eq474202}
  \left\{
  \begin{array}{l}
    \beta+1=\left(5i+2\right)\beta^3,\\[2mm]
    \left(5i+2\right)\beta^2\left(\beta+1\right)=5j+4.\\
  \end{array}
  \right.
  \end{equation}
  This leads to $\left(5i+2\right)^2=5j+4$.
  Since it can be verified that $p\nmid \left(5i+2\right)$, it follows from $c^{\left(2\right)}\left(1\right)=0$ that
  \begin{equation}\label{eq474203}
    \beta^2=\left(5i+2\right)\left(5i+3\right).
  \end{equation}
  Then $\left(\ref{eq474201}\right)$ and $c^{\left(1\right)}\left(1\right)=0$ indicates that
  \begin{equation}\label{eq474204}
    \left(5i+2\right)\beta^2+\beta-\left(5j+3\right)=0.
  \end{equation}
  Let $t=5i+2$. Then one has $\beta+1=t\beta^3$ and $\beta^2=t\left(t+1\right)$ due to the first equality of (\ref{eq474202}) and (\ref{eq474203}).
  It follows from (\ref{eq474204}) that
  \begin{equation*}
    t^2(t+1)+\beta-(t^2-1)=0
  \end{equation*}
  which implies $\beta+1=-t^3$.
  Combining with $\beta+1=t\beta^3$, we have $\beta^3=-t^2$.
  Since $\beta$ is a primitive $5$-th root of unity, one can derive
  \begin{equation*}
  \begin{split}
  0=&\,\beta^4+\beta^3+\beta^2+\beta+1
  \\
  =&\left(\beta+1\right)\left(\beta^3+1\right)+\beta^2
  \\
  =&-t^3\left(-t^2+1\right)+t\left(t+1\right)
  \\
  =&\,t\left(t+1\right)\left(t^3-t^2+1\right).
  \end{split}
  \end{equation*}
  It follows from $t\left(t+1\right)=\beta^2\neq 0$ that $t^3-t^2+1=0$.
  Then we obtain
  \begin{equation*}
    \beta=-t^3-1=-t^2=\beta^3
  \end{equation*}
  which yields $\beta^2-1=0$, a contradiction.

  (3) For the subcase $c(x)=1+a_{1}\,x+a_{2}\,x^{5i+3}+a_{3}\,x^{5j+4}$ with $0\leq i<j\leq p-2$ and $a_{1},\,a_{2},\,a_{3}\in \mathbb{F}_p^{*}$, it follows from $c\left(1\right)=c\left(\beta\right)=c\left(\beta^2\right)=0$ and Lemma \ref{lemsolution} that
  \begin{equation}\label{eq474301}
    a_{1}=\frac{\beta^2}{\beta+1},\quad a_{2}=-\frac{1}{\beta+1},\quad a_{3}=-\beta.
  \end{equation}
  Since $c^{\left(1\right)}\left(1\right)=c^{\left(1\right)}\left(\beta^2\right)=0$, one can immediately get
  \begin{equation*}
  \left\{
  \begin{array}{l}
    \left(5i+3\right)\left(\beta^4-1\right)a_{2}+\left(5j+4\right)\left(\beta-1\right)a_{3}=0,\\[2mm]
    a_{1}(\beta-1)=\left(5i+3\right)\left(\beta^4-\beta\right)a_{2}.\\
  \end{array}
  \right.
  \end{equation*}
  Together with $\left(\ref{eq474301}\right)$, one can conclude that
  \begin{equation*}
  \left\{
  \begin{array}{l}
    \left(5i+3\right)\left(\beta^2+1\right)+(5j+4)\beta=0,\\[2mm]
    \left(5i+3\right)\left(\beta^2+\beta+1\right)+\beta=0.\\
  \end{array}
  \right.
  \end{equation*}
  which indicates
  \begin{equation*}
  \left\{
  \begin{array}{l}
    \left(5i+3\right)\beta^2+(5j+4)\beta+5i+3=0,\\[2mm]
    \left(5i+3\right)\beta^2+(5i+4)\beta+5i+3=0.\\
  \end{array}
  \right.
  \end{equation*}
  It follows that $5(i-j)=0$, a contradiction.

{\bf Case III}:\,\,$(w_H(c(x)), \,w_p(c(x)))=(5,\,7)$.
In this case, we can assume that $c(x)$ is of the form
\begin{equation*}
  \left(\mathbf{a},\,\mathbf{0},\,\mathbf{b},\,\mathbf{0}\right)
\end{equation*}
where $\mathbf{a}$, $\mathbf{b}$ are row vectors with all entries of $\mathbf{a}$, $\mathbf{b}$ being nonzero.
Then its certain cyclic shift must have the form
\begin{equation*}
  \left(\star,\,\star,\,\star,\,\star,\,\mathbf{0},\,\star,\,\mathbf{0}\right)
\end{equation*}
or
\begin{equation*}
  \left(\star,\,\star,\,\star,\,\mathbf{0},\,\star,\,\star,\,\mathbf{0}\right).
\end{equation*}

\begin{itemize}
  \item
  For $\left(\star,\,\star,\,\star,\,\star,\,\mathbf{0},\,\star,\,\mathbf{0}\right)$, there are five subcases to be considered:

  (1) Consider the subcase $c(x)=1+a_{1}\,x+a_{2}\,x^{2}+a_{3}\,x^3+a_{4}\,x^{5i}$ with $1\leq i\leq p-1$ and $a_{1},\,a_{2},\,a_{3},\,a_{4}\in \mathbb{F}_{p}^{*}$.
  It can be verified that
  \begin{equation*}
  \left\{
  \begin{array}{l}
    1+a_{1}+a_{2}+a_{3}+a_{4}=0,\\[2mm]
    1+a_{1}\,\beta+a_{2}\,\beta^2+a_{3}\,\beta^3+a_{4}=0,\\[2mm]
    1+a_{1}\,\beta^2+a_{2}\,\beta^4+a_{3}\,\beta+a_{4}=0\\
  \end{array}
  \right.
  \end{equation*}
  since $c\left(1\right)=c\left(\beta\right)=c\left(\beta^2\right)=0$.
  Then one can derive that $p\nmid \left(a_{4}+1\right)$.
  By Lemma \ref{lemsolution}, one can obtain
  \begin{equation}\label{eq574101}
    a_{1}=-\frac{\beta^2+\beta+1}{\beta^2}(a_{4}+1),
    a_{2}=\frac{\beta^2+\beta+1}{\beta^3}(a_{4}+1),
    a_{3}=-\frac{1}{\beta^3}(a_{4}+1).
  \end{equation}
  It follows from $c^{\left(1\right)}\left(1\right)=c^{\left(1\right)}\left(\beta^2\right)=0$ that
  \begin{equation*}
  \left\{
  \begin{array}{l}
    a_{1}+2\,a_{2}+3\,a_{3}+5i\,a_{4}=0,\\[2mm]
    a_{1}+2\,a_{2}\,\beta^2+3\,a_{3}\,\beta^4+5i\,a_{4}\,\beta^3=0\\
  \end{array}
  \right.
  \end{equation*}
  which indicates
  \begin{equation*}
    \left(\beta^3-1\right)a_{1}+2\left(\beta^3-\beta^2\right)a_{2}
    +3\left(\beta^3-\beta^4\right)a_{3}=0.
  \end{equation*}
  Combining with $\left(\ref{eq574101}\right)$, one can derive that
  \begin{equation*}
    -\left(\beta^3-1\right)\beta\left(\beta^2+\beta+1\right)+2\beta^2\left(\beta-1\right)\left(\beta^2+\beta+1\right)
    +3\beta^3\left(\beta-1\right)=0.
  \end{equation*}
  Since $\beta$ is a primitive $5$-th root of unity, by expanding the above equality, one can get
  $\beta^2+3\beta+1=0$.
  This is contradictory with the inequality (\ref{eqbelta}) in Lemma \ref{lemsolution}.

  (2) Consider the subcase $c(x)=1+a_{1}\,x+a_{2}\,x^2+a_{3}\,x^3+a_{4}\,x^{5i+1}$ with $1\leq i\leq p-1$ and $a_{1},\,a_{2},\,a_{3},\,a_{4}\in \mathbb{F}_{p}^{*}$.
  It follows from $c\left(1\right)=c\left(\beta\right)=c\left(\beta^2\right)=0$ and Lemma \ref{lemsolution} that
  \begin{equation}\label{eq574102}
    a_{1}+a_{4}=-\frac{\beta^2+\beta+1}{\beta^2},\quad
    a_{2}=\frac{\beta^2+\beta+1}{\beta^3},\quad
    a_{3}=-\frac{1}{\beta^{3}}.
  \end{equation}
  Then $c^{\left(1\right)}\left(1\right)=c^{\left(1\right)}\left(\beta^2\right)=0$ induces that
  \begin{equation*}
  \left\{
  \begin{array}{l}
    a_{1}+2\,a_{2}+3\,a_{3}+\left(5i+1\right)a_{4}=0,\\[2mm]
    a_{1}+2\,a_{2}\,\beta^2+3\,a_{3}\,\beta^4+\left(5i+1\right)a_{4}=0.\\
  \end{array}
  \right.
  \end{equation*}
  This leads to
  \begin{equation*}
    2\left(\beta^2-1\right)a_{2}+3\left(\beta^4-1\right)a_{3}=0.
  \end{equation*}
  Together with $\left(\ref{eq574102}\right)$, one can immediately get $\left(\beta-1\right)^2=0$, which is impossible.

  (3) Consider the subcase $c(x)=1+a_{1}\,x+a_{2}\,x^2+a_{3}\,x^3+a_{4}\,x^{5i+2}$ with $1\leq i\leq p-1$ and $a_{1},\,a_{2},\,a_{3},\,a_{4}\in \mathbb{F}_{p}^{*}$.
  The fact $c\left(1\right)=c\left(\beta\right)=c\left(\beta^2\right)=0$ and Lemma \ref{lemsolution} induces
  \begin{equation}\label{eq574103}
    a_{1}=-\frac{\beta^2+\beta+1}{\beta^2},\quad
    a_{2}+a_{4}=\frac{\beta^2+\beta+1}{\beta^3},\quad
    a_{3}=-\frac{1}{\beta^3}.
  \end{equation}
  It follows from $c^{\left(1\right)}\left(1\right)=c^{\left(1\right)}\left(\beta^2\right)=0$ that
  \begin{equation*}
  \left\{
  \begin{array}{l}
    a_{1}+2\,a_{2}+3\,a_{3}+\left(5i+2\right)a_{4}=0,\\[2mm]
    a_{1}+2\,a_{2}\,\beta^2+3\,a_{3}\,\beta^4+\left(5i+2\right)a_{4}\,\beta^2=0\\
  \end{array}
  \right.
  \end{equation*}
  which implies
  \begin{equation*}
    \left(\beta^2-1\right)a_{1}+3\left(\beta^2-\beta^4\right)a_{3}=0.
  \end{equation*}
  By substituting $\left(\ref{eq574103}\right)$ into the above equality, we have $\left(\beta-1\right)^2=0$, a contradiction.

  (4) Consider the subcase $c(x)=1+a_{1}\,x+a_{2}\,x^2+a_{3}\,x^3+a_{4}\,x^{5i+3}$ with $1\leq i\leq p-1$ and $a_{1},\,a_{2},\,a_{3},\,a_{4}\in \mathbb{F}_{p}^{*}$.
  By $c\left(1\right)=c\left(\beta\right)=c\left(\beta^2\right)=0$ and Lemma \ref{lemsolution}, one has
  \begin{equation}\label{eq574104}
    a_{1}=-\frac{\beta^2+\beta+1}{\beta^2},\quad
    a_{2}=\frac{\beta^2+\beta+1}{\beta^3},\quad
    a_{3}+a_{4}=-\frac{1}{\beta^3}.
  \end{equation}
  It follows from $c^{\left(1\right)}\left(1\right)=c^{\left(1\right)}\left(\beta^2\right)=0$ that
  \begin{equation*}
  \left\{
  \begin{array}{l}
    a_{1}+2\,a_{2}+3\,a_{3}+\left(5i+3\right)a_{4}=0,\\[2mm]
    a_{1}+2\,a_{2}\,\beta^2+3\,a_{3}\,\beta^4+\left(5i+3\right)a_{4}\,\beta^4=0.\\
  \end{array}
  \right.
  \end{equation*}
  This yields
  \begin{equation*}
    \left(\beta^4-1\right)a_{1}+2\left(\beta^4-\beta^2\right)a_{2}=0.
  \end{equation*}
  Combining with $\left(\ref{eq574104}\right)$, one can derive that $\left(\beta-1\right)^2=0$, which is impossible.

  (5) Consider the subcase $c(x)=1+a_{1}\,x+a_{2}\,x^2+a_{3}\,x^3+a_{4}\,x^{5i+4}$ with $1\leq i\leq p-2$ and $a_{1},\,a_{2},\,a_{3},\,a_{4}\in \mathbb{F}_{p}^{*}$.
  It can be verified that
  \begin{equation*}
  \left\{
  \begin{array}{l}
    1+a_{1}+a_{2}+a_{3}+a_{4}=0,\\[2mm]
    1+a_{1}\,\beta+a_{2}\,\beta^2+a_{3}\,\beta^3+a_{4}\,\beta^4=0,\\[2mm]
    1+a_{1}\,\beta^2+a_{2}\,\beta^4+a_{3}\,\beta+a_{4}\,\beta^3=0\\
  \end{array}
  \right.
  \end{equation*}
  since $c\left(1\right)=c\left(\beta\right)=c\left(\beta^2\right)=0$.
  Then one can obtain that
  \begin{equation}\label{eq574105}
  \left\{
  \begin{array}{l}
    a_{1}=-\beta^3\,a_{4}+\beta^2+\beta,\\[2mm]
    a_{2}=-\left(\beta^4+1\right)a_{4}-\beta-1,\\[2mm]
    a_{3}=-\left(\beta^2+\beta+1\right)a_{4}-\beta^2.\\
  \end{array}
  \right.
  \end{equation}
  It follows from $c^{\left(1\right)}\left(1\right)=c^{\left(1\right)}\left(\beta^2\right)=0$ that
  \begin{equation*}
  \left\{
  \begin{array}{l}
    a_{1}+2\,a_{2}+3\,a_{3}+\left(5i+4\right)a_{4}=0,\\[2mm]
    a_{1}+2\,a_{2}\,\beta^2+3\,a_{3}\,\beta^4+\left(5i+4\right)a_{4}\,\beta=0\\
  \end{array}
  \right.
  \end{equation*}
  which implies
  \begin{equation*}
    \left(\beta-1\right)a_{1}+2\left(\beta-\beta^2\right)a_{2}
    +3\left(\beta-\beta^4\right)a_{3}=0.
  \end{equation*}
  This is equivalent to
  \begin{equation*}
    a_1-2\,\beta\,a_2-3\,\beta\left(\beta^2+\beta+1\right)a_3=0.
  \end{equation*}
  Together with $\left(\ref{eq574105}\right)$, one can immediately have
  \begin{equation*}
    \left(-\beta^3+2\beta(\beta^4+1)+3\,\beta\left(\beta^2+\beta+1\right)^2\right)a_4+\beta^2+\beta+2\,\beta(\beta+1)
    +3\,\beta^3\left(\beta^2+\beta+1\right)=0.
  \end{equation*}
  Then we get that
  \begin{equation*}
    -\beta^3+2\beta\left(\beta^4+1\right)+3\,\beta\left(\beta^2+\beta+1\right)^2=0
  \end{equation*}
  due to $\beta^4+\beta^3+\beta^2+\beta+1=0$ and $a_4\in \mathbb{F}_{p}^{*}$.
  By a straightforward computation, one has $\beta^2+3\beta+1=0$.
  This contradicts with the inequality (\ref{eqbelta}) in Lemma \ref{lemsolution}.

  \item
  For $\left(\star,\,\star,\,\star,\,\mathbf{0},\,\star,\,\star,\,\mathbf{0}\right)$, there are also five subcases to be considered:

  (1) Consider the subcase $c(x)=1+a_{1}\,x+a_{2}\,x^2+a_{3}\,x^{5i}+a_{4}\,x^{5i+1}$ with $1\leq i\leq p-1$ and $a_{1},\,a_{2},\,a_{3},\,a_{4}\in \mathbb{F}_{p}^{*}$.
  It follows from $c\left(1\right)=c\left(\beta\right)=c\left(\beta^2\right)=0$ that
  \begin{equation*}
  \left\{
  \begin{array}{l}
    1+a_{1}+a_{2}+a_{3}+a_{4}=0,\\[2mm]
    1+a_{1}\,\beta+a_{2}\,\beta^2+a_{3}+a_{4}\,\beta=0,\\[2mm]
    1+a_{1}\,\beta^2+a_{2}\,\beta^4+a_{3}+a_{4}\,\beta^2=0\\
  \end{array}
  \right.
  \end{equation*}
  which implies
  \begin{equation*}
  \left\{
  \begin{array}{l}
    \left(a_{1}+a_{4}\right)\left(\beta-1\right)
    +a_{2}\left(\beta^2-1\right)=0,\\[2mm]
    \left(a_{1}+a_{4}\right)\left(\beta^2-\beta\right)
    +a_{2}\left(\beta^4-\beta^2\right)=0.\\
  \end{array}
  \right.
  \end{equation*}
  This indicates that $\beta\left(\beta^2-1\right)a_{2}=\left(\beta^4-\beta^2\right)a_{2}$.
  Hence $\beta=1$, a contradiction.

  (2) Consider the subcase $c(x)=1+a_{1}\,x+a_{2}\,x^2+a_{3}\,x^{5i+1}+a_{4}\,x^{5i+2}$ with $1\leq i\leq p-1$ and $a_{1},\,a_{2},\,a_{3},\,a_{4}\in \mathbb{F}_{p}^{*}$.
  It can be verified that
  \begin{equation*}
  \left\{
  \begin{array}{l}
    1+a_{1}+a_{2}+a_{3}+a_{4}=0,\\[2mm]
    1+a_{1}\,\beta+a_{2}\,\beta^2+a_{3}\,\beta+a_{4}\,\beta^2=0,\\[2mm]
    1+a_{1}\,\beta^2+a_{2}\,\beta^4+a_{3}\,\beta^2+a_{4}\,\beta^4=0\\
  \end{array}
  \right.
  \end{equation*}
  since $c\left(1\right)=c\left(\beta\right)=c\left(\beta^2\right)=0$.
  Then one can derive that
  \begin{equation*}
  \left\{
  \begin{array}{l}
    \left(a_{2}+a_{4}\right)\left(\beta^2-\beta\right)=\beta-1,\\[2mm]
    \left(a_{2}+a_{4}\right)\left(\beta^4-\beta^3\right)=\beta-1.\\
  \end{array}
  \right.
  \end{equation*}
  It follows that $\beta^3=\beta$, which is impossible.

  (3) Consider the subcase $c(x)=1+a_{1}\,x+a_{2}\,x^2+a_{3}\,x^{5i+2}+a_{4}\,x^{5i+3}$ with $1\leq i\leq p-1$ and $a_{1},\,a_{2},\,a_{3},\,a_{4}\in \mathbb{F}_{p}^{*}$.
  The fact $c\left(1\right)=c\left(\beta\right)=c\left(\beta^2\right)=0$ and Lemma \ref{lemsolution} induces that
  \begin{equation}\label{eq573203}
    a_{1}=-\frac{\beta^2+\beta+1}{\beta^2},\quad
    a_{2}+a_{3}=\frac{\beta^2+\beta+1}{\beta^3},\quad
    a_{4}=-\frac{1}{\beta^3}.
  \end{equation}
  It follows from $c^{\left(1\right)}\left(1\right)=c^{\left(1\right)}\left(\beta^2\right)=0$ that
  \begin{equation*}
  \left\{
  \begin{array}{l}
    a_{1}+2\,a_{2}+\left(5i+2\right)a_{3}+\left(5i+3\right)a_{4}=0,\\[2mm]
    a_{1}+2\,a_{2}\,\beta^2+\left(5i+2\right)a_{3}\,\beta^2
    +\left(5i+3\right)a_{4}\,\beta^4=0.\\
  \end{array}
  \right.
  \end{equation*}
  This yields
  \begin{equation*}
    \left(\beta^2-1\right)a_{1}+\left(5i+3\right)\left(\beta^2-\beta^4\right)a_{4}=0.
  \end{equation*}
  By substituting $\left(\ref{eq573203}\right)$, one can deduce that
  \begin{equation*}
    \beta^2-\left(5i+2\right)\beta+1=0.
  \end{equation*}
  Let $t=5i+2$.
  Then $\beta^2=t\beta-1$ and
  \begin{equation*}
    \beta^4+\beta^3+\beta^2+\beta+1=(t\beta-1)(t^2+t-1)=0.
  \end{equation*}
  It follows that $t^2+t=1$.
  By $c^{\left(2\right)}\left(1\right)=0$ and (\ref{eq573203}), we get
  \begin{equation*}
    5i\left(t+1\right)a_{3}=\left(t+2\right)\beta+1.
  \end{equation*}
  The fact $c^{\left(1\right)}\left(1\right)=0$ indicates $5i\,a_{3}=\left(2-t\right)\left(\beta+1\right)$.
  Hence
  \begin{equation*}
    \left(t+2\right)\beta+1=\left(t+1\right)\left(2-t\right)\left(\beta+1\right).
  \end{equation*}
  This leads to $t^2\,\beta-2\,t=0$ due to $t^2+t=1$.
  It follows from $t\ne 0$ that $t\beta=2$ and $\beta^2=t\beta-1=1$, a contradiction.

  (4) Consider the subcase $c(x)=1+a_{1}\,x+a_{2}\,x^2+a_{3}\,x^{5i+3}+a_{4}\,x^{5i+4}$ with $1\leq i\leq p-2$ and $a_{1},\,a_{2},\,a_{3},\,a_{4}\in \mathbb{F}_{p}^{*}$.
  It can be checked that
  \begin{equation*}
  \left\{
  \begin{array}{l}
    1+a_{1}+a_{2}+a_{3}+a_{4}=0,\\[2mm]
    1+a_{1}\,\beta+a_{2}\,\beta^2+a_{3}\,\beta^3+a_{4}\,\beta^4=0,\\[2mm]
    1+a_{1}\,\beta^2+a_{2}\,\beta^4+a_{3}\,\beta+a_{4}\,\beta^3=0\\
  \end{array}
  \right.
  \end{equation*}
  since $c\left(1\right)=c\left(\beta\right)=c\left(\beta^2\right)=0$.
  Then one can derive that
  \begin{equation}\label{eq573204}
  \left\{
  \begin{array}{l}
    a_{1}=-\beta^3\,a_{4}+\beta^2+\beta,\\[2mm]
    a_{2}=-\left(\beta^4+1\right)a_{4}-\beta-1,\\[2mm]
    a_{3}=-\left(\beta^2+\beta+1\right)a_{4}-\beta^2.\\
  \end{array}
  \right.
  \end{equation}
  Let $t=5i+2$.
  By $c^{\left(1\right)}\left(1\right)=c^{\left(1\right)}\left(\beta^2\right)=0$ and (\ref{eq573204}), we have
  \begin{equation}\label{eq57320401}
  \left\{
  \begin{array}{l}
    t\beta^2+\beta+2=\left((t-1)\beta^4+t\beta^3+t\right)a_{4},\\[2mm]
    2\beta^2+\beta+t=\left(t\beta^2+(t-1)\beta+t\right)a_{4}.\\
  \end{array}
  \right.
  \end{equation}
  Then
  \begin{equation*}
    (t\beta^2+\beta+2)\left(t\beta^2+(t-1)\beta+t\right)=(2\beta^2+\beta+t)\left((t-1)\beta^4+t\beta^3+t\right)
  \end{equation*}
  which implies
  \begin{equation*}
    (t^2+t-1)(\beta^2-1)=0.
  \end{equation*}
  Thus $t^2+t=1$.
  It follows from $c^{\left(2\right)}\left(1\right)=0$ that $2\,a_{2}+a_{3}+\left(2t+3\right)a_{4}=0$.
  Together with $\left(\ref{eq573204}\right)$, one can immediately get
  \begin{equation*}
    \left(-\beta^4+\beta^3+2t+1\right)a_{4}=\beta^2+2\beta+2.
  \end{equation*}
  Combining with the second equality in (\ref{eq57320401}), we can obtain
  \begin{equation*}
    \left(-\beta^4+\beta^3+2t+1\right)\left(2\beta^2+\beta+t\right)
    =\left(\beta^2+2\beta+2\right)\left(t\beta^2+(t-1)\beta+t\right).
  \end{equation*}
  By expanding the above equality, one can deduce
  \begin{equation*}
    \left(\beta^2-1\right)t+3\beta^2+2=0
  \end{equation*}
  which yields $t=\frac{3\beta^2+2}{1-\beta^2}$.
  The fact $t^2+t-1=0$ induces
  \begin{equation*}
    \left(\frac{3\beta^2+2}{1-\beta^2}\right)^2+\frac{3\beta^2+2}{1-\beta^2}-1=0
  \end{equation*}
  which is equivalent to
  \begin{equation*}
    \left(3\beta^2+2\right)^2+\left(3\beta^2+2\right)\left(1-\beta^2\right)-\left(1-\beta^2\right)^2=0.
  \end{equation*}
  It follows that
  \begin{equation*}
    \beta^4+3\,\beta^2+1=0
  \end{equation*}
  which indicates
  \begin{equation*}
    2\,\beta^2-\beta^3-\beta=0
  \end{equation*}
  due to $\beta^4+\beta^3+\beta^2+\beta+1=0$.
  Hence $\beta(\beta-1)^2=0$, which is impossible.

  (5) Consider the subcase $c(x)=1+a_{1}\,x+a_{2}\,x^2+a_{3}\,x^{5i+4}+a_{4}\,x^{5i+5}$ with $0\leq i\leq p-2$ and $a_{1},\,a_{2},\,a_{3},\,a_{4}\in \mathbb{F}_{p}^{*}$.
  It follows from $c\left(1\right)=c\left(\beta\right)=c\left(\beta^2\right)=0$ that $p\nmid \left(a_{4}+1\right)$ and
  \begin{equation}\label{eq573205}
    a_{1}=-\frac{1}{\beta}\left(a_{4}+1\right),\quad
    a_{2}=-\frac{\beta}{\beta+1}\left(a_{4}+1\right),\quad
    a_{3}=\frac{1}{\beta(\beta+1)}\left(a_{4}+1\right)
  \end{equation}
  due to Lemma \ref{lemsolution}.
  The fact $c^{\left(1\right)}\left(1\right)=c^{\left(1\right)}\left(\beta^2\right)=0$ leads to
  \begin{equation*}
  \left\{
  \begin{array}{l}
    a_{1}+2\,a_{2}+\left(5i+4\right)a_{3}+\left(5i+5\right)a_{4}=0,\\[2mm]
    a_{1}+2\,a_{2}\,\beta^2+\left(5i+4\right)a_{3}\,\beta
    +\left(5i+5\right)a_{4}\,\beta^3=0\\
  \end{array}
  \right.
  \end{equation*}
  which implies
  \begin{equation*}
  \left\{
  \begin{array}{l}
    \left(\beta^3-1\right)a_{1}+2\left(\beta^3-\beta^2\right)a_{2}
    +\left(5i+4\right)\left(\beta^3-\beta\right)a_{3}=0,\\[2mm]
    2\left(\beta+1\right)a_{2}+\left(5i+4\right)a_{3}+\left(5i+5\right)\left(\beta^2+\beta+1\right)a_{4}=0.\\
  \end{array}
  \right.
  \end{equation*}
  By substituting $\left(\ref{eq573205}\right)$, one can obtain that
  \begin{equation}\label{eq57320501}
    \beta^4+\left(5i+3\right)\beta^3-\left(5i+3\right)\beta-1=0
  \end{equation}
  and
  \begin{equation}\label{eq57320502}
    \left(-2\beta^2\left(\beta+1\right)+5i+4\right)\left(a_{4}+1\right)
    +\left(5i+5\right)\beta\left(\beta+1\right)\left(\beta^2+\beta+1\right)a_4=0.
  \end{equation}
  Let $t=5i+3$.
  It follows from (\ref{eq57320501}) that
  \begin{equation*}
    \beta^4-1+t\left(\beta^3-\beta\right)=\left(\beta^2-1\right)(\beta^2+1+t\,\beta)=0
  \end{equation*}
  which yields $\beta^2=-t\beta-1$.
  Then we have
  \begin{equation*}
    0=\beta^4+\beta^3+\beta^2+\beta+1=-\left(t^2-t-1\right)\beta^2
  \end{equation*}
  which indicates $t^2=t+1$ due to $\beta^2\ne 0$.
  It can be verified that
  \begin{equation*}
  \begin{split}
   &-2\beta^2\left(\beta+1\right)+5i+4+\left(5i+5\right)\beta\left(\beta+1\right)\left(\beta^2+\beta+1\right)
  \\
  =&-2\beta^3-2\beta^2+t+1-\left(t+2\right)\left(\beta^4+\beta+2\right)
  \\
  =&-2t\left(\beta+1\right)+2\left(t\beta+1\right)+\left(t+2\right)\left(\beta+t\right)-\left(t+2\right)\beta-t-3=0.
  \end{split}
  \end{equation*}
  Hence (\ref{eq57320502}) and $a_4\in \mathbb{F}_{p}^{*}$ induces
  \begin{equation*}
    0=-2\beta^2\left(\beta+1\right)+5i+4=3-t
  \end{equation*}
  which means that $t=3$ and $\beta^2=-3\beta-1$, a contradiction with the inequality (\ref{eqbelta}) in Lemma \ref{lemsolution}.
\end{itemize}

As a consequence, $\mathcal{C}$ is an MDS $\left(5p,\,8\right)_{p}$ symbol-pair code. The desired result follows.
\hfill $\Box$

\end{document}